\documentclass{article}
\usepackage{amsmath}
\usepackage{amsthm}
\usepackage{fancybox}
\usepackage{hyperref}
\newcommand{\uft}[1]{\footnote{\url{#1}}}

\begin{document}

\noindent\fbox{%
    \parbox{\textwidth}{%
        \centering
        \Large The Future of MEV \\
        \large \textit{An Analysis of Ethereum Execution Tickets} \\
        \normalsize By: \href{https://www.linkedin.com/in/jonah-burian/}{Jonah Burian} (03/05/2024)
    }
}

\begin{abstract}
This paper analyzes the Execution Tickets proposal\uft{https://ethresear.ch/t/execution-tickets/17944} on Ethereum Research, unveiling its potential to revolutionize the Ethereum blockchain's economic model. At the core of this proposal lies a novel ticketing mechanism poised to redefine how the Ethereum protocol distributes the value associated with proposing execution payloads.\uft{https://eth2book.info/capella/part3/containers/execution/} This innovative approach enables the Ethereum protocol to directly broker Maximal Extractable Value (MEV),\uft{https://ethereum.org/en/developers/docs/mev} traditionally an external revenue stream for validators. The implementation of Execution Tickets goes beyond optimizing validator compensation; it also introduces a new Ethereum native asset with a market capitalization expected to correlate closely with the present value of all value associated with future block production. The analysis demonstrates that the Execution Ticket system can facilitate a more equitable distribution of value within the Ethereum ecosystem, and pave the way for a more secure and economically robust blockchain network.
\end{abstract}

\subsubsection*{Acknowledgments}
\textit{I am grateful to Aleks Larsen, Barnabé Monnot, Davide Crapis, Justin Drake, Martin Rosen, Mike Neuder, and others at Blockchain Capital\uft{https://www.blockchaincapital.com/} and the Ethereum Foundation\uft{https://ethereum.foundation/} for their insightful discussions and thorough reviews.}

\section{Introduction}

\subsection{MEV Dynamics in Blockchain Consensus}
\subsubsection{MEV: Conceptual Overview and Historical Context}
The Ethereum\uft{https://ethereum.org/} blockchain operates using a proof-of-stake consensus mechanism. In this system, a network of nodes, known as validators, maintains the Ethereum Virtual Machine (EVM).\uft{https://ethereum.org/developers/docs/evm} At each slot, a validator is selected to propagate a block, with the validator's probability of being chosen directly proportional to its staked ETH amount. The block contains, among other data, an execution payload, namely a list of ordered transactions. It is noteworthy that there are typically more transactions available to be placed in a block than there is space within the block, and the selection and ordering of these transactions present a combinatorial problem.

With the implementation of EIP-1559,\uft{https://github.com/ethereum/EIPs/blob/master/EIPS/eip-1559.md} Ethereum introduced a dual fee structure comprising base fees and priority fees. While base fees are burned, contributing to Ethereum's deflationary aspect, priority fees are given to validators as an incentive for block production. On the surface, block construction to maximize revenue sounds simple: just choose the transactions with the highest priority fees. However, block construction to maximise revenue is far more complex than this naive strategy because there exists an additional source of revenue: Maximal Extractable Value (MEV). MEV arises from a validator's ability to choose and order transactions. The most prevalent MEV strategies include sandwich attacks, arbitrage, and liquidations. A typical example of a sandwich attack occurs when a validator spots a buy order for inclusion in a block and strategically places its own transactions before and after it (i.e., front-running and back-running the order). Arbitrage might involve exploiting price differences between a Centralized Exchange (CEX) and a Decentralized Exchange (DEX). 

\subsubsection{Execution Layer Rewards}
Execution Layer Rewards (EL Rewards) encapsulate the total value a validator can earn from proposing an execution payload. This term comprehensively includes the revenue from two primary sources: fees + MEV.

\subsubsection{Implications of MEV on Validator Economics}
Over the years, MEV has evolved into a complex landscape, reminiscent of high frequency trading, where various entities engage in high-pressure, time-sensitive, combinatorially complex economic games. Recognizing that validators may not be best equipped for these intricate games, Flashbots\uft{https://www.flashbots.net/} introduced a new mechanism called MEV-Boost.\uft{https://docs.flashbots.net/flashbots-mev-boost/introduction} This mechanism, also known as Proposer Builder Separation (PBS),\uft{https://arxiv.org/abs/2305.19037} distinguishes the block-building function from the proposing function. Validators can opt into an auction where builders bid for the right to choose the execution payload. Validators then propose the execution payload with the highest associated bid. Since the builders select the payload, they receive the EL Rewards minus the bid, while the validator receives the bid. Effectively, builders are paying validators for the right to construct an execution payload. Given the competitiveness of the market and the short term monopoly the validator has on block production, the winning bid ends up being a little shy of the value of the EL Rewards.

\subsection{Reassessing the Distribution of EL Rewards}
This situation brings to the forefront the issue of how the total value of EL Rewards, encompassing both MEV and transaction fees, is distributed. Validators currently benefit from the full spectrum of these rewards, but this raises questions about the optimal allocation of value within the Ethereum ecosystem. Is it more beneficial for the network's long-term health and security if a portion of these EL Rewards is captured and redirected to the protocol itself? This would require somehow brokering MEV payments at the protocol level.

The potential redistribution of EL Rewards could lead to a more equitable and balanced economic model within Ethereum, enhancing network security and aligning with its deflationary goals. The challenge lies in determining an effective mechanism for capturing a part of the EL Rewards without undermining the incentives necessary for validators.

One approach to better distribute EL Rewards would be to enshrine PBS,\uft{https://ethresear.ch/t/why-enshrine-proposer-builder-separation-a-viable-path-to-epbs/15710} also known as ePBS, at the protocol level and to burn the bid using a mechanism called MEV-Burn.\uft{https://ethresear.ch/t/mev-burn-a-simple-design/15590} While this proposal has been hotly debated, an orthogonal approach has surfaced. This new proposal will be the focus of this paper.

\subsection{The Execution Tickets Mechanism}
Drake and Neuder have proposed a new MEV capture method that creates two distinct roles with separate reward systems and an auction-based selection process.\uft{https://ethresear.ch/t/execution-tickets/17944}

\subsubsection{Separation of Duties in Ethereum's Network}
\begin{itemize}
    \item \textbf{Beacon Block Proposer}: 

        This proposer role is similar to that of current validators, except the Beacon Block Proposer does not construct the execution payload. Instead, it selects a subset of transactions that must be included. The Beacon Block Proposer is chosen randomly and proportionally to its staked amount of ETH just like before.

    \item \textbf{Execution Block Proposer}: 
    
        This proposer is tasked with constructing the execution payload and constrained by the requirement to include the transactions in the inclusion list. Selection occurs through an innovative auction system involving the sale and drawing of tickets. The likelihood of winning the lottery is equal to the number of tickets one holds over the total tickets. Winning the lottery confers the right to propose an execution payload and capture the associated EL Rewards. The winner can still outsource the construction.

\end{itemize}

\subsubsection{Operational Framework of the Lottery System}
The way the Execution Block Proposer lottery operates is as follows: When the protocol activates, \( n \) lottery tickets are minted and can be purchased from the protocol. For every block, one ticket is drawn, determining the Execution Block Proposer. After the winner proposes the execution payload, the winner's ticket is burned and a new one is minted which is available for anyone to buy. This process continues indefinitely.

\subsubsection{Economic Implications of the Lottery on Ethereum's Value}
These tickets not only confer the right to EL Rewards but also introduce a new native asset to Ethereum, potentially creating their own economy and market. Revenue from ticket sales is intended to be burned, applying deflationary pressure on ETH's supply, increasing ETH's value, and thereby enhancing network security. In essence, burning MEV contributes to Ethereum's security budget. The exact mechanism for selling these tickets is still an open research question (as the protocol will most likely only be able to capture the value of the tickets from primary, not secondary, sales).

\section{Motivation and Objectives}
This paper aims to delve into the Execution Block Proposer mechanism, examining its efficacy in capturing EL Rewards and assessing the potential market value of these tickets. It will also assess what the market cap of these tickets reveals about the MEV landscape. Finally, the paper will pose questions for further exploration. This paper will not explore how these tickets should be priced to capture their full value.

\section{Theoretical Framework}
\subsection{Modeling the MEV Ticket Economy}
\begin{itemize}
    \item \textbf{Scenario Description}: The game starts with \( n \) tickets. A ticket provides an opportunity to win a prize (EL Reward, i.e., MEV + fees) at discrete time intervals (slot), denoted by \( t \). At each increment, 1 ticket wins. Winning results in the ticket being voided (burned) and excluded from subsequent draws. A single fresh ticket is generated (minted) following the award of the prize at each interval.
    \item \textbf{Winning Probability}: The probability of a single ticket winning the prize at each time interval \( t \) is \( \frac{1}{n} \).
    \item \textbf{Reward at Time \( t \)}: Victory at time $t$ yields a prize with a value represented by the random variable $\mathcal{R}$.

    Note: In practice, the distribution of EL Rewards (\( \mathcal{R} \)) may change over time, and the values \( r_i \) and \( r_j \) across intervals might be dependent.\footnote{If a block contains a significant amount of MEV, this might suggest that the following block could also have a significant amount of MEV, e.g., when rates are cut and DEXs are mispriced.} For simplicity, it is assumed that \( \mathcal{R} \) has a distribution that does not vary with time and that each draw is independent.

    \item \textbf{Discount Factor}: The value of the prize at time interval \( t \) is discounted to its present value using a discount factor of \( \frac{1}{(1+d)^t} \), where \( d \) represents the inter-slot discount rate. 
    For simplicity, it is assumed that \( d \) does not change with time.
    \item \textbf{Present Value}: The function \( \textit{PVal}(X,t) = \frac{X}{(1+d)^t} \)  will be used to represent the present value of some value X realized at time \(t\). For example, the present value of the EL Reward at time interval \( t \) is given by \( \textit{PVal}(\mathcal{R},t) = \frac{\mathcal{R}}{(1+d)^t} \).
    \item \textbf{Present Value of a ticket}: \( V_{ticket} \) denotes the Net Present Value (NPV) of a single ticket, which is also a random variable.
\end{itemize}

\section{Analytical Exploration}
\subsection{Valuation of Future EL Reward Streams}

    \textbf{Expected Net Present Value of the EL Rewards}: As \( \mathcal{R} \)  is a random variable, so is its present value \( \textit{PVal}(\mathcal{R},t) \). The expected present value is:
    \[ E[\textit{PVal}(\mathcal{R},t)] = \frac{E[\mathcal{R}]}{(1+d)^t} = \frac{\mu_{\mathcal{R}}}{(1+d)^t} = \textit{PVal}(\mu_{\mathcal{R}},t) \]
    Where \( \mu_{\mathcal{R}} \) is the expected value of \( \mathcal{R} \).

    \newtheorem{theorem}{Theorem}
    \begin{theorem}
    \label{theorem:NPV}
    The expected net present value of all future EL Rewards, 
    \[ NPV_{\mathcal{R}} = \frac{\mu_{\mathcal{R}}}{d} \]
    \end{theorem}
    \begin{proof}
        \begin{align*}
        NPV_{\mathcal{R}} &= E\Bigl[\sum_{t=1}^{\infty}\textit{PVal}(\mathcal{R},t)\Bigr] \\
        &= \sum_{t=1}^{\infty} E[\textit{PVal}(\mathcal{R},t)] \\
        &=\sum_{t=1}^{\infty} \textit{PVal}(\mu_{\mathcal{R}},t) \\
        &= \sum_{t=1}^{\infty} \frac{\mu_{\mathcal{R}}}{(1+d)^t} \\
        &= \mu_{\mathcal{R}}\sum_{t=1}^{\infty} \frac{1}{(1+d)^t} \\
        &= \frac{\mu_{\mathcal{R}}}{d}
        \end{align*} 
        
        Note:  \(\sum_{t=1}^{\infty} \frac{1}{(1+d)^t} = \frac{1}{d}\) because \(|1+d| > 1 \) (this a modified geometric series).
 
    \end{proof}

\subsection{Expected Value of an Individual Tickets}

\begin{theorem}
    The expected value of a single ticket, 
    \[ E[V_{ticket}] =  \frac{\mu_{\mathcal{R}}}{nd+1} \]
\end{theorem}

\begin{proof}
Using the Law of Total Expectations:
\begin{align*}
 E[V_{ticket}] = \sum_{t=1}^{\infty}  P(W=t) \cdot E[V_{ticket}|W=t]
\end{align*} where \(W  =t \) means the ticket win at time t. \\

Note that the probability of a ticket winning at time \(t\) is \(\frac{1}{n}\). Additionally, the probability that a ticket has not won through time \(t-1\) is given by \(\left(1 - \frac{1}{n}\right)^{t-1}\).\\

Therefore,
\begin{align*}
    E[V_{ticket}] &= \sum_{t=1}^{\infty} P(W=t) \cdot E[V_{ticket}|W=t] \\
    &= \sum_{t=1}^{\infty} \left(1 - \frac{1}{n}\right)^{t-1} \cdot \frac{1}{n} \cdot \textit{PVal}(\mu_{\mathcal{R}},t) \\
    &= \sum_{t=1}^{\infty}  \left(1 - \frac{1}{n}\right)^{t-1} \cdot \frac{1}{n} \cdot \frac{\mu_{\mathcal{R}}}{(1+d)^t} \\
    &= \frac{\mu_{\mathcal{R}}}{n(1+d)}  \cdot \sum_{t=1}^{\infty}  \left(1 - \frac{1}{n}\right)^{t-1} \cdot  \frac{1}{(1+d)^{t-1}} \\
    &= \frac{\mu_{\mathcal{R}}}{n(1+d)}  \cdot \sum_{t=1}^{\infty}  \left(\frac{n-1}{n(1+d)}\right)^{t-1} \\
    &= \frac{\mu_{\mathcal{R}}}{n(1+d)} \left(\frac{1}{1-(\frac{n-1}{n(1+d)})}\right)\\
    &= \frac{\mu_{\mathcal{R}}}{n(1+d)} \left(\frac{n(1+d)}{nd+1}\right)\\
    &= \frac{\mu_{\mathcal{R}}}{nd+1}\\
\end{align*}

Note: \(\sum_{t=1}^{\infty}  \left(\frac{n-1}{n(1+d)}\right)^{t-1} 
= \frac{1}{1-(\frac{n-1}{n(1+d)})} \) because \(|\frac{n-1}{n(1+d)}| < 1\) (this is a modified geometric series).
\end{proof}

\subsection{Aggregate Valuation of Issued and Unissued Tickets}
\subsubsection{Synthesizing Individual and Market-Wide Valuations}
Given the closed-form solution for the value of a ticket, the net present value of all tickets issued and unissued can be calculated.

\begin{theorem}
    The net present value of all tickets issued and unissued equals the expected net present value of the rewards.
    \[ E[V_{\text{all tickets}}] = NPV_{\mathcal{R}}\]
\end{theorem}

\begin{proof}
\begin{align*}
    E[V_{\text{all tickets}}] &=  E[V_{\text{issued tickets}}] +  E[V_{\text{unissued tickets}}] \\
    &=  nE[V_{ticket}] +  \sum_{t=1}^{\infty} \textit{PVal}(E[V_{ticket}],t) \\
    &=  nE[V_{ticket}] +  \sum_{t=1}^{\infty} \frac{E[V_{ticket}]}{(1+d)^t} \\
    &=  nE[V_{ticket}] +  \frac{E[V_{ticket}]}{d} \\
    &=  E[V_{ticket}] \left(\frac{nd + 1}{d}\right) \\
    &=  \left(\frac{\mu_{\mathcal{R}}}{nd+1}\right) \left(\frac{nd+1}{d}\right) \\
    &= \frac{\mu_{\mathcal{R}}}{d} \\
    &= NPV_{\mathcal{R}}
\end{align*}
\end{proof}
Therefore, assuming an efficient market, Execution Tickets capture all \( \mathcal{R} \) in the system in expectation.

\subsubsection{Market Efficiency and EL Reward Capture}
A critical aspect of the system's design involves determining the optimal quantity of tickets to issue (the optimal \( n \)) and assessing its impact on the value captured by the system from the outset.

\begin{theorem}
\label{theorem:large-n}
If n is sufficiently large, the current market cap of all tickets equals the present value of all future EL Rewards
\end{theorem}

\begin{proof}
\begin{align*}
     \lim_{n \to \infty} E[V_{\text{issued tickets}}] &=  \lim_{n \to \infty}  nE[V_{ticket}] \\
     &= \lim_{n \to \infty} \left(\frac{n\mu_{\mathcal{R}}}{nd+1}\right)\\
     &= \lim_{n \to \infty}  \left(\frac{\mu_{\mathcal{R}}}{d+\frac{1}{n}}\right)\\
     &=  \frac{\mu_{\mathcal{R}}}{d} \\
    &=  NPV_{\mathcal{R}}
\end{align*}
\end{proof}

Effectively, tickets are EL Reward futures and the market cap of the tickets is a leading indicator of the future value that can be captured from proposing the execution payload.

\subsubsection{Impact of Ticket Quantity on Ticket Values}
While a large \( n \) is analytically appealing, it increases the expected timeframe until an individual ticket wins, which actually decreases the value of an individual ticket.

\begin{theorem}
\label{theorem:time-n}
The expected number of slots until a ticket wins is \( n \)
\end{theorem}

\begin{proof}
Define \( T \) as a random variable denoting the number of slots required to win the lottery. \( T \) follows a geometric distribution since each slot represents an independent Bernoulli trial with a constant success probability \( p = \frac{1}{n} \).

\[ E(T) = \frac{1}{p} = \frac{1}{\frac{1}{n}} = n \]
\end{proof}

\begin{theorem}
\label{theorem:price-n}
The expected value of a ticket decreases with \( n \)
\end{theorem}
\begin{proof}
\begin{align*}
     \frac{\partial}{\partial n} E[V_{\text{ticket}}] &= \frac{\partial}{\partial n} \left(\frac{\mu_{\mathcal{R}}}{nd+1}\right) \\
     &= - \frac{\mu_{\mathcal{R}}d}{(nd+1)^2} < 0
\end{align*}
\end{proof}

\subsection{Cost of Control}
\begin{theorem}
\label{theorem:cost-n}
The value of commanding \(P\%\) of the outstanding tickets increases in \( n \)
\end{theorem}
\begin{proof}
Controlling \(P\%\) of the tickets involves holding \(pn\) tickets, where \(p = \frac{P}{100}\). Thus, the expected value for controlling \(P\%\) of the tickets is:
\begin{align*}
E[V_{\text{P\% of tickets}}] &= pnE[V_{\text{ticket}}] \\
&=\frac{pn\mu_{\mathcal{R}}}{nd+1}
\end{align*}

Differentiating with respect to \(n\) yields:
\begin{align*}
     \frac{\partial}{\partial n} E[V_{\text{P\% of tickets}}] &= \frac{\partial}{\partial n}\left( \frac{pn\mu_{\mathcal{R}}}{nd+1} \right) \\
     &= \frac{p\mu_{\mathcal{R}}}{(nd+1)^2} > 0
\end{align*}
\end{proof}

\section{Challenges and Limitations}

\subsection{Multi-Block MEV and Centralization Risks}
One of the most critical challenges to this study relates to the phenomenon of multi-block MEV. This concept posits that controlling multiple consecutive blocks can yield disproportionately higher rewards than the sum of the individual blocks' values.\footnote{With multi-block MEV, the value from controlling  slots \(i\) and \(i+1\) is greater than  \(r_i + \frac{r_{i+1}}{1+d}\). This means that \(E[V_{\text{P\% of tickets}}] > pnE[V_{\text{ticket}}] \) as P increases.
} If multi-block MEV is prevalent (though its existence in practice is not yet clear), the advantage gained from controlling successive blocks could lead to significant, potentially exponential, centralization pressures. This trend would be contrary to Ethereum's ethos of decentralization and could pose a threat to network security.

Moreover, the introduction of Execution Tickets, while innovative, may not fully address this issue. Multi-block MEV could lead to the formation of centralized ticket pools, with entities controlling a large number of consecutive blocks. This could further amplify concerns regarding network centralization. This complex dynamic suggests that if multi-block MEV does exist, it would require an entirely different model than the one proposed above to accurately factor in the greater-than-the-sum-of-its-parts effect inherent in multi-block control. The implications of such a scenario are existential and necessitate further investigation.

\subsection{Pricing Mechanisms Complexities}
The effectiveness of the execution ticket system hinges critically on the method employed for setting ticket prices. It is imperative to consider the implicit discount that will inevitably factor into the pricing strategy. This discount reflects the operational costs, risks, and the required profit margins for execution proposers. As a result, the total value obtained from the sale of tickets is unlikely to equate precisely to the EL Rewards, i.e., \(P_{ticket} < E[V_{ticket}]\) or, said differently, \(E[V_{ticket}] - P_{ticket} = \textit{Required Profit}\) where \(P_{ticket}\) is the price of a ticket.

Therefore, the pricing mechanism must be designed thoughtfully to capture as much value as possible for the protocol while still being attractive and viable for proposers. This delicate balance is key to ensuring that the protocol benefits maximally from ticket sales. An optimal pricing strategy would aim to minimize the gap between the expected value of the tickets and the practical selling price, which would reduce the potential value leakage into secondary markets and benefit arbitrageurs over the Ethereum protocol.

\subsection{Valuation Complexities}
Valuing tickets will be challenging due to the necessity of forecasting over an infinite time horizon. This issue is further complicated by the lumpiness of EL Rewards, i.e., the high variance of \( \mathcal{R} \), which may change over time (contrary to the assumption). This variance contributes to the variance in the value of the ticket as \(\text{Var}(V_{ticket})\) scales with \(\text{Var}(\mathcal{R})\).\footnote{See ~\ref{theorem:var-ticket}} Additionally, \(d\) will not remain static (contrary to the assumption), necessitating forecasts to account for events like rate hikes and cuts. 

Such complexity inevitably leads to the mispricing of tickets, raising concerns about whether this ticketing mechanism is purely theoretical and not suitable for practice. The primary concern lies in underpricing, which could suggest a failure to capture the intended EL Rewards. Conversely, overpricing (potentially more likely due to the speculative or gambling premium associated with the tickets) could result in capturing excess value (a scenario that might not be entirely negative since it redistributes rewards back to the protocol).

The emergence of ticket pooling can mitigate the impact of \( R \)'s variance. By distributing rewards and variance across a pool, the system can reduce pricing complexity and could lead to more accurate ticket valuations. More work is needed to formalize the impact of pooling.

\subsection{Challenges in Choosing the Parameter \(n\)}

The selection of \(n\), representing both the initial number of tickets issued and the ongoing number of outstanding tickets per block, emerges as a critical design parameter with significant implications for the system's dynamics and economic outcomes. This parameter directly influences the amount of value captured upfront by the protocol, the prices of the tickets, the cost of monopolizing the supply, and the simplicity of interpreting the market capitalization.

\subsubsection{Trade-offs of a Larger \(n\)}

A larger \(n\) enhances security by increasing the cost of owning a significant share of tickets as described in Theorem~\ref{theorem:cost-n}. This effectively protects against monopolizing block construction rights and increases the cost of centralization. Moreover, per Theorem~\ref{theorem:price-n}, a larger \(n\) results in a lower price per ticket, democratizing the ability to win the right to propose the execution payload. Less critically, a larger \(n\) holds vanity value as observed in Theorem~\ref{theorem:large-n}, which posits that, as \(n\) grows sufficiently large, the current market cap of all tickets approximates the present value of all future EL Rewards. Such a simplification is appealing.

However, the intricacies of managing a large \(n\) introduce several challenges. Primarily, a larger \(n\) prolongs the expected timeframe for any individual ticket to win per Theorem~\ref{theorem:time-n}, complicating the valuation and forecasting efforts for participants. This is particularly problematic given the assumption that \(\mu_{\mathcal{R}}\), the expected MEV, remains constant over time—an assumption that may not hold, further amplifying forecasting challenges. This uncertainty can impose a significant additional upfront discount on ticket values, allowing the secondary market to capture a large chunk of the ticket value. 

\subsubsection{Trade-offs of a Smaller \(n\)}

Conversely, opting for a smaller \(n\) shifts more value to later sales, providing market participants with more information and potentially stabilizing the ticket market, thereby enhancing primary market value capture. This approach, however, raises concerns regarding ticket accessibility and network centralization. According to Theorem~\ref{theorem:price-n}, a smaller \(n\) results in a higher price per ticket, potentially excluding average participants from direct participation in the ticket lottery due to financial constraints. More importantly, controlling a large percentage of the network would be less expensive according to Theorem~\ref{theorem:cost-n}, raising concerns about centralization.

\subsubsection{Implications of Pricing Mechanisms on \(n\)}

Different pricing mechanisms may be employed for the initial sale of \(n\) tickets and subsequent per-block ticket sales, each with varying efficiencies in capturing the true value of a ticket. A more effective initial sale mechanism may justify a larger \(n\), ensuring more value is captured upfront. In contrast, a superior per-block sale mechanism could justify a smaller \(n\), facilitating an effective value capture over the long term.

\section{Related Work}

This section contrasts the proposed Execution Tickets mechanism with established concepts like MEV-Burn and MEV-Smoothing, and highlights the unique advantages of the  Execution Tickets approach.

\subsection{Comparison with MEV-Burn}

\subsubsection{Simplification}
Neuder notes: ``The current version of mev-burn is tightly coupled with the attesting committee for a given slot. The burning mechanism in the execution ticket design is more straightforward..."\uft{https://ethresear.ch/t/execution-tickets/17944} 

\subsubsection{Timing and Value Capture}
In the simple MEV-Burn model,\uft{https://ethresear.ch/t/mev-burn-a-simple-design/15590} the amount burned equals the highest bid in the first \(D\) seconds of the auction. It is known that MEV bids increase monotonically over time (see ``timing games"\uft{https://ethresear.ch/t/timing-games-implications-and-possible-mitigations/17612}) because builders have more time to explore the problem space. Consequently, MEV-Burn fails to burn all potential MEV. The Execution Ticket model, conversely, captures the full MEV value, albeit in expectation.

\subsubsection{Credible Neutrality}
Another issue with MEV-Burn is its reliance on ePBS, and the problem lies in the very name---PBS is enshrined. This approach is opinionated and not credibly neutral.\uft{https://vitalik.eth.limo/general/2023/09/30/enshrinement.html} Better supply chains and auction systems might be discovered in the future. The ticketing system is agnostic to the method of MEV capture and simply presupposes that MEV value will flow to the execution proposer due to the proposer's inherent block space monopoly.

\subsubsection{Bypassability Issues}
A significant challenge with MEV-Burn is the ePBS bypassability problem.\uft{https://notes.ethereum.org/@mikeneuder/infinite-buffet} There is no credible method to ensure all proposers and builders use the protocol. The ticketing system avoids this issue by not dictating an MEV capture method. This flexibility renders the ticketing system less prone to circumvention compared to ePBS.

\subsection{Comparison with MEV-Smoothing}
MEV-Smoothing\uft{https://notes.ethereum.org/cA3EzpNvRBStk1JFLzW8qg} is a proposal designed to evenly distribute MEV captured across validators, aiming to reduce the variability in their earnings. This approach raises important philosophical questions about the extent to which a protocol should be involved in shaping reward distribution.

\subsubsection{Philosophical Considerations on Reward Distribution}
A fair argument can be made that the protocol should not intervene in artificially smoothing rewards, but allow for a free market where individuals and entities decide their own tolerance for variances in earnings. Those seeking lower variance in the proposed Execution Tickets mechanism can join ticket pools akin to the current staking pools offered by Lido\uft{https://lido.fi/} and Coinbase.\uft{https://www.coinbase.com/cbeth} However, it is important to consider the potential for centralization (like what is happening with liquid staking pools). Further research is necessary to determine the viability of decentralized ticket pools.

\section{Conclusion and Discussion}

This paper has investigated the potential of Ethereum's novel Execution Tickets mechanism for capturing EL Rewards within the protocol itself. The findings indicate that, when priced correctly, the tickets can indeed internalize all value generated from proposing execution payloads and redirect what was once validator revenue to the protocol. This is a significant breakthrough. It suggests that the contentious issue of MEV capture could potentially be resolved without necessitating ePBS.

The analysis reveals that a ticket represents a clean and elegant abstraction: a share of of future EL Rewards. Such a mechanism not only simplifies the economic model of network value distribution but also sets the foundation for a robust market where tickets serve as leading indicators of network value generation. Moreover, as these tickets are sold prior to a blocks construction, the revenue generated could bolster Ethereum's security budget in anticipation of high volatility events.

However, the success of this mechanism hinges on the efficacy of the ticket sale process. The protocol must be capable of selling tickets at their intrinsic value; otherwise, the value would inevitably leak into a secondary market.

The proposed lottery ticket system represents a promising approach to MEV management on the Ethereum network. By potentially funneling MEV back into the protocol, it aligns incentives, enhances network security, and creates a new asset class within the Ethereum ecosystem. While further research is required to refine the ticket sales mechanism and assess long-term impacts, the initial findings are undoubtedly encouraging, and signal a potential paradigm shift in the way blockchain protocols can harness MEV for systemic benefit.

\appendix
\section{Appendix}

\subsection{Variance in Value of an Individual Ticket}
\label{theorem:var-ticket}
    \begin{theorem}
    The variance in the value of a single ticket,
    
    \begin{align*}
    \text{Var}(V_{ticket}) &= \frac{\text{Var}(\mathcal{R})+\mu_{\mathcal{R}}^2}{nd^2+2nd+1} - \frac{\mu_{\mathcal{R}}^2}{n^2d^2 + 2nd + 1}
    \end{align*}
    
    \end{theorem}

 \begin{proof}
    \begin{align*}
     \text{Var}(V_{ticket}) &= E[V_{ticket}^2] - E[V_{ticket}]^2 \\
     &= E[V_{ticket}^2] - \left(\mu_{\mathcal{R}} \left(\frac{1}{nd+1}\right)\right)^2 \\
     &= \frac{\text{Var}(\mathcal{R})+\mu_{\mathcal{R}}^2}{nd^2+2nd+1} - \frac{\mu_{\mathcal{R}}^2}{n^2d^2 + 2nd + 1}
    \end{align*}   

    Since,
\begin{align*}
    E[V_{ticket}^2] &= \sum_{t=1}^{\infty} P(W=t) \cdot E[V_{ticket}^2|W=t] \\
    &= \sum_{t=1}^{\infty}  \left(1 - \frac{1}{n}\right)^{t-1} \cdot \frac{1}{n} \cdot E\left[\left(\frac{R}{(1+d)^t}\right)^2\right]\\
    &= \sum_{t=1}^{\infty}  \left(1 - \frac{1}{n}\right)^{t-1} \cdot \frac{1}{n} \cdot \left(\frac{E[R^2]}{(1+d)^{2t}}\right)\\
    &= \sum_{t=1}^{\infty}  \left(1 - \frac{1}{n}\right)^{t-1} \cdot \frac{1}{n} \cdot \left(\frac{\text{Var}(\mathcal{R})+E[R]^2}{(1+d)^{2t}}\right)\\
    &= \sum_{t=1}^{\infty}  \left(1 - \frac{1}{n}\right)^{t-1} \cdot \frac{1}{n} \cdot \left(\frac{\text{Var}(\mathcal{R})+\mu_{\mathcal{R}}^2}{(1+d)^{2t}}\right)\\
    &= \frac{\text{Var}(\mathcal{R})+\mu_{\mathcal{R}}^2}{n(1+d)^2} \cdot \sum_{t=1}^{\infty}  \left(1 - \frac{1}{n}\right)^{t-1} \cdot \left(\frac{1}{(1+d)^{2t-2}}\right)\\
    &= \frac{\text{Var}(\mathcal{R})+\mu_{\mathcal{R}}^2}{n(1+d)^2} \cdot
    \sum_{t=1}^{\infty}\left(\frac{1-\frac{1}{n}}{(1+d)^2}\right)^{t-1}\\
    &= \frac{\text{Var}(\mathcal{R})+\mu_{\mathcal{R}}^2}{n(1+d)^2} \cdot \sum_{t=1}^{\infty}\left(\frac{n-1}{n(1+d)^2}\right)^{t-1} \\
    \end{align*}   
    \pagebreak
    \begin{align*}
    &= \frac{\text{Var}(\mathcal{R})+\mu_{\mathcal{R}}^2}{n(1+d)^2}\left(\frac{1}{1-\left(\frac{n-1}{n(1+d)^2}\right)}\right)\\
    &= \frac{\text{Var}(\mathcal{R})+\mu_{\mathcal{R}}^2}{n(1+d)^2}\left(\frac{n(1+d)^2}{n(1+d)^2 - (n-1)}\right) \\
    &= \frac{\text{Var}(\mathcal{R})+\mu_{\mathcal{R}}^2}{nd^2+2nd+1}
\end{align*}

Note: \(\sum_{t=1}^{\infty}\left(\frac{n-1}{n(1+d)^2}\right)^{t-1} 
= \frac{1}{1-\left(\frac{n-1}{n(1+d)^2}\right)} \) because \(|\frac{n-1}{n(1+d)^2}| < 1\) (this is a modified geometric series). 
 \end{proof}

\end{document}